\documentclass[12pt]{article}

\usepackage[margin=1in]{geometry}  
\usepackage{graphicx}              
\usepackage{amsmath}               
\usepackage{amsfonts}              
\usepackage{amsthm}                
\usepackage{amssymb}
\usepackage{algorithm}
\usepackage{algorithmic}

\newtheorem{thm}{Theorem}[section]
\newtheorem{lem}[thm]{Lemma}
\newtheorem{prop}[thm]{Proposition}
\newtheorem{cor}[thm]{Corollary}

\newtheorem{defn}{Definition}[section]

\DeclareMathOperator*{\argmax}{arg\,max}

\newcommand{\ua}[1]{\left\uparrow{#1}\right.}
\newcommand{\da}[1]{\left\downarrow{#1}\right.}

\begin{document}

\nocite{*}
\title{Dualization in Lattices Given by Ordered Sets of Irreducibles}

\author{Mikhail A. Babin, Sergei O. Kuznetsov}

\maketitle

\begin{abstract}
Dualization of a monotone Boolean  function on a finite lattice can be represented by transforming the
set of its minimal 1 values to the set of its maximal 0 values.
In this paper we consider finite lattices  given
by ordered sets of their meet and join irreducibles (i.e., as a concept
lattice of a formal context). We show that in this case dualization is equivalent to
the enumeration of so-called minimal hypotheses. In contrast to usual dualization setting, where
a lattice is given by the ordered set of its elements, dualization in this case
is shown to be impossible in output polynomial time unless P = NP. However, if the lattice is
distributive, dualization is shown to be possible in subexponential time.
\end{abstract}

\section{Introduction}

A monotone Boolean function on a finite lattice can be given by the set of
minimal 1 values or by the set of its maximal 0 values.
Dualization is the transformation of the set of minimal 1 values of a Boolean function to the set of its maximal 0 values or vice versa.
Since dualization is equivalent to many important problems in computer and data sciences~\cite{emg08,np12,elb02}, the  paper~\cite{fk96}
on quasi-polynomial dualization algorithm for Boolean lattices was an important breakthrough. It paved the way to generalizations to various classes of structures where dualization in output subexponential time is possible, among them dualization on lattices given by ordered sets of their elements or by products of bounded width lattices, like chains~\cite{emg08,elb02}.

A well-known fact is that every lattice is determined up to isomorphism by the ordered set of its meet (infimum) and join (supremum) irreducible elements~\cite{gw}. These elements cannot be represented as meets (joins) of other elements that are larger (smaller) then them.
On diagram of finite lattices these elements have one upper (lower) neighbor. In this paper we consider finite lattices  given
by ordered sets of their meet and join irreducibles, known as concept lattices~\cite{gw,dp,g11}.
We show that dualization for representation of this type is
impossible  in output polynomial time unless P = NP. However, in an important particular case where the lattice is distributive, we propose a subexponential algorithm.

Dualization in the considered case is not only of theoretical interest. Actually, this study was motivated by a practical problem of enumerating minimal hypotheses, which is a problem of learning specific type of classifiers from positive and negative examples. Hypotheses
or JSM-hypotheses were proposed by V.K.Finn~\cite{f83,f91} and formalized in terms of Formal Concept Analysis (FCA) in~\cite{kuz96,kzg,kuz04}.
The set of minimal hypotheses is classification equivalent to the set of all hypotheses, thus making a  condensed representation of the latter. The set of all hypotheses can be generated with polynomial delay~\cite{kuz04}, however, the problem of
generating minimal hypotheses with polynomial delay remained an open one for long time.
In this paper we show that dualization on lattices given by the ordered set of its irreducible elements is equivalent to enumeration of minimal hypotheses, thus complexity results concerning minimal hypotheses and dualization can be mutually translated.

In what follows we shall use the notation of Formal Concept Analysis~\cite{gw}, which provides a convenient language and necessary results for lattices given by ordered sets of irreducible elements.

The rest of the paper is organized as follows: In the second section we give
most important definitions. In the third section we prove the main
intractability result on impossibility of enumerating minimal hypotheses and dualization
in output polynomial time unless P = NP.
In the fourth section we conclude by
discussing the implication of the results for the problem of dualizing monotone
Boolean functions.
In the fifth section we relate minimum implication base problem to dualization over product of lattices that are given explicitly, and dualization over distributive lattice.
In the sixth section we describe subexponential dualization algorithm for the distributive lattice case.

\subsection{Related work}

To the best of our knowledge all dualization problems that have been studied in previous works consider dualization
over product of posets $\mathcal{P} = \mathcal{P}_1\times\ldots\times\mathcal{P}_k$,
where each poset $\mathcal{P}_i$ is some special type of a poset that is given explicitly. In \cite{elb02,el09} the author give quasi-polynomial time algorithms for the following cases: each $\mathcal{P}_i$ is a join semi-lattice of bounded width (any antichain has constant size), each $\mathcal{P}_i$ is a forest poset in which either the in-degree or the out-degree
of each element is constant (see also \cite{elb022}), each $\mathcal{P}_i$ is the lattice of intervals defined by a set of intervals on the real line $\mathbb{R}$. In \cite{elb02,el09} a more general dualization problem was stated where each $\mathcal{P}_i$ is a lattice (with no bounds on its width), the existence of quasi-polynomial time algorithms for this case is still an open question. In this paper we prove an upper bound complexity of the latter problem via another long-standing open complexity problem, the minimum implication base (see \cite{dg86}, equivalently SID problem from \cite{kh95, emg08}). The most common technique leading to quasi-polynomial time algorithm for duality problems are based on the idea of high frequency based decomposition, first introduced in \cite{fk96}. We use this method to get subexponential algorithm for the dualization over distributive lattice.

Although product of lattices $\mathcal{L}=\mathcal{L}_1\times\ldots\times\mathcal{L}_k$, where each $\mathcal{L}_i$ is given explicitly, can provide
exponentially smaller description of $\mathcal{L}$ not every lattice can have a nontrivial exponentially smaller representation of this kind.

\section{Preliminaries}

\begin{defn}\label{D0} \footnote{We use capital characters to denote elements of partially ordered sets since it agrees with FCA notation for concept lattices.}
A subset $\mathcal{A} \subseteq \mathcal{P}$ of a partially ordered set $(\mathcal{P}, <)$ is called an \emph{antichain} iff $A \nleq B$ for any $A,B \in \mathcal{A}$, i.e.,
all elements of an antichain are incomparable.
\end{defn}

The following property is required in dualization problems.
For two antichains $\mathcal{A}, \mathcal{B} \subseteq \mathcal{P}$
we say $(\mathcal{A},\mathcal{B})$ has property~$(*)$ if
\begin{center}  $A \not\leq B$  for any $A \in \mathcal{A},\ B \in \mathcal{B}\ (*)$. \end{center}

\begin{defn}
Antichains $\mathcal{A},\mathcal{B}\subseteq \mathcal{P}$ of partially ordered set $\mathcal{P}$ are called dual iff
$\mathcal{A}, \mathcal{B}$ satisfy property (*) and for any $P \in \mathcal{P}$ either $P\leq B$ for some $B \in \mathcal{B}$ or
$A \leq P$ for some $A \in \mathcal{A}$.
\end{defn}

The dualization problem over partially ordered set usually have the following statement:

\noindent{\bf Problem:}\ \rm Dualization over partially ordered set $\mathcal{P}$
\newline\noindent\emph{INPUT:}\ \rm Partially ordered set $\mathcal{P}$ (that can be given implicitly),
antichain $\mathcal{A} \subseteq \mathcal{P}$.
\newline\noindent\emph{OUTPUT:}\ \rm Antichain $\mathcal{B}\subseteq \mathcal{P}$ such that $\mathcal{A}$ and $\mathcal{B}$ are dual.

Note that the output $\mathcal{B}$ of the dualization problem can be exponential in the input size $(|\mathcal{A}| \times |[description\ of\ \mathcal{P}]|)$. Therefore, we are interested in the time complexity of dualization that
depends on both input and output sizes. We say that dualization problem can be solved in \emph{output polynomial} time if
there is an algorithm that can generate set $\mathcal{B}$ in time polynomial of
$|\mathcal{B}|\times|\mathcal{A}| \times |[description\ of\ \mathcal{P}]|$.
Usually we will consider decision version of the
dualization problem called \emph{duality} problem:

\noindent{\bf Problem:}\ \rm Duality over partially ordered set $\mathcal{P}$
\newline\noindent\emph{INPUT:}\ \rm Partially ordered set $\mathcal{P}$ (that can be given implicitly),
antichains $\mathcal{A},\mathcal{B} \subseteq \mathcal{P}$ satisfying (*).
\newline\noindent\emph{QUESTION:}\ \rm Are antichains $\mathcal{A}$ and $\mathcal{B}$ dual?

Equivalent definition of the dualization over poset can be given using monotone Boolean\footnote{Hereafter by Boolean functions we mean Boolean-valued functions.} functions
on a partially ordered set. Let $f: \mathcal{P} \mapsto \{0,1\}$ be a monotone Boolean function on a partially ordered set $\mathcal{P}$, i.e. $X \leq Y \Rightarrow f(X) \leq f(Y)$ and $\mathcal{A}$ is a set of minimal 1-values of $f$. Clearly, the set of maximal 0-values of $f$ is dual to $\mathcal{A}$.

In this paper we consider only the case where the partially ordered set over which we dualize is a lattice.
\label{sect:basics}
A partial ordered set $(\mathcal{L}, <)$ is called a \emph{lattice}~\cite{dp} if any
pair of its elements has an infimum (meet $\wedge$) and a supremum (join $\vee$). Equivalently,
a lattice is an algebra $(\mathcal{L},\wedge,\vee)$ with the following properties of
$\wedge$ and $\vee$:
\begin{itemize}

\item[L1] $X\vee X  = X$, \quad $X\wedge X  = X$ \ (idempotence)

\item[L2] $X\vee Y  = Y\vee X$, \quad $X\wedge Y  = Y\wedge X$ \ (commutativity)

\item[L3] $X\vee (Y\vee Z)  = (X\vee Y)\vee Z$, \quad $X\wedge (Y\wedge Z)  = (X\wedge Y)\wedge Z$ \ (associativity)

\item[L4] $X = X\wedge (X\vee Y) = X\vee (X\wedge Y)$ \ (absorption)

\end{itemize}
A lattice is called \emph{complete} if every subset of it has infimum and supremum.

A lattice is distributive if for any $X,Y,Z\in \mathcal{L}$
$$X\wedge (Y\vee Z) = (X\wedge Y) \vee (X\wedge Z).$$
The following elements of a lattice are very important in our work. An element $X\in \mathcal{L}$
is called \emph{infimum-irreducible} (or \emph{meet-irreducible})
if $X \neq \bigwedge_{Y > X} Y$, i.e., $X$ is not represented by the intersection of any elements above it. Dually,
an element $X\in \mathcal{L}$ is called \emph{supremum-irreducible} (or \emph{join-irreducible})
if $X \neq \bigvee_{Y < X} Y$, i.e., $X$ is not represented by the union of any elements below it.
Meet- (join-) irreducible elements have only one upper (lower) neighbor in the lattice diagram.

In what follows we use the standard definitions and facts of Formal Concept Analysis (FCA) from~\cite{gw}.
Let $G$ and $M$ be sets, called the set of \emph{objects} and \emph{attributes},
respectively. Let $I$ be a relation $I\subseteq G \times M$ between
objects and attributes: for $g\in G, m\in M, gIm$ holds iff the
object $g$ has the attribute $m$. The triple $\mathbb{K}=(G, M, I)$
is called a \emph{(formal) context} and is naturally represented by a cross-table, where
rows stay for objects, columns stay for attributes and crosses stay for pairs $(g,m)\in I$.
If $A\subseteq G, B\subseteq M$
are arbitrary subsets, then the following \emph{derivation operators}
$$A'=\{m\in M\mid gIm \ \forall g \in A \} $$
$$B'=\{g\in G\mid gIm \ \forall m \in B \} $$

define \emph{Galois connection} between ordered powersets $(2^G,\subseteq)$ and $(2^M,\subseteq)$, 
since $A\subseteq B' \iff B\subseteq A'$.
The pair $(A,B)$, where $A\subseteq G$, $B\subseteq M$, $A'=B$, and
$B'=A$ is called a \emph{(formal) concept (of the context
$\mathbb{K}$)} with \emph{extent} $A$ and \emph{intent} $B$ (in this
case we have also $A''=A$ and $B''=B$). Formal concepts are ordered by the
following relation
$$(A_1,B_1)\leq (A_2,B_2) \mbox{ iff}\ A_1\subseteq A_2 (B_2\subseteq B_1),$$
this partial order being a complete lattice on the set of all concepts.
This lattice is called a \emph{concept lattice} ${\cal L}(G,M,I)$ of the context $(G,M,I)$.

The set of join-irreducible elements of a concept lattice ${\cal L}(G,M,I)$
is contained in the set of
\emph{object concepts}, which have the form $(g'', g')$, $g\in G$. Dually, the set of meet-irreducible elements
of a concept lattice is contained in the set of
\emph{attribute concepts}, which have the form $(m', m'')$, $m\in M$.
An object $g$ is called \emph{reducible} if $g' = M$ or $\exists X \subseteq G\setminus \{g\}:\ g' = \bigcap \limits_{j\in X} j'$, i.e., the respective row of the context cross-table is either full or is an intersection of some other rows. If $g$ is not reducible, then $(g'',g')$ is a join-irreducible element of ${\cal L}(G,M,I)$. Dually, an attribute $m$ is called \emph{reducible} if $m' = G$ or $\exists Y \subseteq M\setminus \{m\}:\ m' = \bigcap \limits_{j\in Y} j'$, i.e. the respective column of the context cross-table is either full or is an intersection of some other columns. If $m$ is not reducible, then $(m',m'')$ is a meet-irreducible element of ${\cal L}(G,M,I)$.

The Basic Theorem of FCA~\cite{gw} implies that every finite lattice $(L,\vee, \wedge)$ can be represented as a concept lattice
${\cal L}(J(L),M(L),\leq)$, where $J(L)$ is the set of all join-irreducible elements of $L$, $M(L)$ is the set of meet-irreducible elements of $L$, and $\leq$
is the natural partial order of $(L, \vee, \wedge)$.

A set of attributes $B$ is
\emph{implied} by a set of attributes $A$, or implication
$A\rightarrow B$ holds, if all objects from $G$ that have all
attributes from $A$ also have all attributes from
$B$, i.e. $A'\subseteq B'$. Implications obey Armstrong rules
$${{}\over{X\to X}}\quad , \quad
        {{X\to Y}\over{X\cup Z\to Y}}\quad , \quad
        {{X\to Y, Y\cup Z\to W}\over{X\cup Z\to W}},$$
and a minimal subset of implications from which all other implications
can be deduced by means of Armstrong rules is called an \emph{implication base}.
In~\cite{dg86} a characterization of cardinality-minimum implication base (Duquenne-Guigues base) was given.
\par\medskip

\section{Enumeration of minimal hypotheses}
\label{sec:enumminhyp}
Now we present a learning model from~\cite{f83,f91}  in terms of FCA~\cite{kuz96,kzg,kuz04}.
This model complies with the common paradigm of learning from
positive and negative examples (see, e.g. \cite{kzg}, \cite{kuz04} ):
given a positive and negative examples of a ``target attribute",
construct a generalization of the positive examples that would not
cover any negative example.

Let $t$ be  \emph{target}  attribute, different
from attributes from the set $M$, which correspond to
\emph{structural} attributes of objects. For example, in
pharmacological applications the structural attributes can
correspond to particular subgraphs of molecular graphs of chemical
compounds.

Input data for learning can be represented by sets of positive,
negative, and undetermined examples. \emph{Positive examples} (or
$(+)$-examples) are objects that are known to have the target attribute $t$
and \emph{negative examples} (or $(-)$-examples) are objects that
are known not to have this attribute.

\begin{defn}\label{D1}
Consider positive context $\mathbb{K}_+=(G_+,M,\mathcal{I}_+)$ and
negative context $\mathbb{K}_-=(G_-,M,\mathcal{I}_-)$. The context
$\mathbb{K}_{\pm}=(G_+\cup G_-,M \cup \{w\},\mathcal{I}_+\cup
\mathcal{I_-}\cup G_+\times \{w\})$ is called a \emph{training
context}. The derivation operators in these contexts are denoted by
superscripts $(\cdot)^+$, $(\cdot)^-$, and $(\cdot)^\pm$, respectively.
\end{defn}

\begin{defn}\label{D2}
A subset $H \subseteq M$ is called a positive (or $(+)$-)-hypothesis
of training context $\mathbb{K}_{\pm}$ if $H$ is intent of
$\mathbb{K}_+$ and $H$ is not a subset of any intent of
$\mathbb{K}_-$. For $k\in N\cup \{0\}$ a subset $H \subseteq M$ is
called a $k$-weak positive (or $k(+)$-)-hypothesis
of training context $\mathbb{K}_{\pm}$ if $H$ is intent of
$\mathbb{K}_+$ and $|H^+\cap G_-|\leq k$.

\end{defn}
Obviously, a positive hypothesis is a 0-weak hypothesis. Weak hypotheses stay for noise-tolerant dependencies, which are important in data mining applications.
In the same way negative (or $(-)$-) hypotheses are defined.

Besides classified objects (positive and negative examples), one
usually has objects for which the value of the target attribute is unknown. These
examples are usually called undetermined examples, they can be given
by a context $\mathbb{K}_\tau:=(G_\tau,M,I_\tau),$ where  the
corresponding derivation operator is denoted by  $(\cdot)^\tau$.

Hypotheses can be used to classify the undetermined examples: If the
intent $$g^\tau:=\{m\in M\mid (g,m)\in I_\tau\}$$ of an object $g\in
G_\tau$ contains a positive, but no negative hypothesis, then
$g^\tau$ is {\em classified positively}. Negative classifications
are defined similarly. If $g^\tau$ contains hypotheses of both
kinds, or if $g^\tau$ contains no hypothesis at all, then the
classification is contradictory or undetermined, respectively. In
this case one can apply probabilistic techniques.

In \cite{kzg}, \cite{kuz04} it was argued that one can restrict to {\em minimal}
(w.r.t.\ inclusion $\subseteq$) hypotheses, positive as well as
negative, since an object intent $g^\tau$ obviously contains a positive
hypothesis if and only if it contains a minimal positive hypothesis.

\begin{defn}\label{D3}
For $k\in \mathbb{N}\cup \{0\}$ if the set of $k(+)$-hypotheses is not empty, then
$H$ is a minimal $k(+)$-hypothesis iff $H$ is a $k(+)$-hypothesis and $F$ is not a $k(+)$-hypothesis for any $F\subset H$.
In case the set of $k(+)$-hypotheses is empty, we put the set of minimal $k(+)$-hypotheses consisting of the only set $M$.
\end{defn}

The latter condition is needed technically for dualization: without it not every monotone Boolean function would be
dualizable.

\medskip

\noindent {\bf Example.}  Consider the following training context, where $m_0$ is the target attribute,
the set of attributes is $M=\{m_1,\ldots,m_6\}$, the set of negative examples is $G=\{g_1,g_2,g_3\}$,  the set of positive examples is $G_+=\{g_4,\ldots,g_9\}$ and the incidence relation $I$ is given by the following cross-table:

\medskip
\begin{center}
\begin{tabular}{|c|c|cccccc|}
\hline
$G\setminus M$&$m_0$& $m_1$ &$m_2$ & $m_3$ &$m_4$ &$m_5$ & $m_6$\\
\hline
$g_1$ &  &  & $\times$ & $\times$&  & $\times$ & $\times$\\
$g_2$ &  & $\times$ &  & $\times$& $\times$ &  & $\times$\\
$g_3$ &  & $\times$ & $\times$ & & $\times$ & $\times$ & \\

\hline
$g_4$ & $\times$ &  & $\times$ & $\times$ & $\times$ & $\times$ & $\times$\\
$g_5$ & $\times$ & $\times$ &  & $\times$ & $\times$ & $\times$ & $\times$\\
$g_6$ & $\times$ & $\times$ & $\times$ &  & $\times$ & $\times$ & $\times$\\
$g_7$ & $\times$ & $\times$ & $\times$ & $\times$ &  & $\times$ & $\times$\\
$g_8$ & $\times$ & $\times$ & $\times$ & $\times$ & $\times$ &  & $\times$\\
$g_9$ & $\times$ & $\times$ & $\times$ & $\times$ & $\times$ & $\times$ & \\
\hline
\end{tabular}
\end{center}

\medskip

\noindent Here, we have $2^3 = 8$ minimal hypotheses: $\{m_1, m_2, m_3\}$,
$\{m_1, m_2, m_6\}$, $\{m_1, m_5, m_3\}$, $\{m_1, m_5, m_6\}$,
$\{m_4, m_2, m_3\}$, $\{m_4, m_2, m_6\}$, $\{m_4, m_5, m_3\}$,
$\{m_4, m_5, m_6\}$.

\medskip

In what follows we will also need the following definition from FCA, which is important in
constructing ``hard cases" for FCA-related complexity problems.

\begin{defn}
Let $G=\{g_1,\ldots,g_n\}$ and $M=\{m_1,\ldots,m_n\}$ be sets
with same cardinality. Then the context
$\mathbb{K}=(G,M,\mathcal{I}_{\neq})$ is called \emph{contranominal
scale}, where $\mathcal{I}_{\neq}= G\times M \setminus
\{(g_1,m_1),\ldots,(g_n,m_n)\}$.
\end{defn}

The contranominal scale has the following property, which we will
use later: for any $H\subseteq M$ one has $H''=H$ and $H'=\{g_i\mid
m_i\notin H, 1\leq i\leq n\}$.

Here we discuss algorithmic complexity of enumerating all minimal
hypotheses. Note that there is an obvious algorithm for enumerating
all hypotheses (not necessary minimal) with polynomial
delay~\cite{kuz04}. This algorithm is an adaptation of an algorithm
for computing the set of all concepts, where the branching condition is changed to include the additional
condition $|H^+\cap G_-|\leq k$.

\noindent{\bf Problem:}\ \rm Minimal hypotheses enumeration (MHE)
\newline\noindent\emph{INPUT:}\ \rm Positive and negative contexts
$\mathbb{K}_+=(G_+,M,\mathcal{I}_+),\mathbb{K}_-=(G_-,M,\mathcal{I}_-)$
\newline\noindent\emph{OUTPUT:}\ \rm All minimal hypotheses of $\mathbb{K}_{\pm}$.

Unfortunately, this problem cannot be solved in output polynomial
time unless $P=NP$. In order to prove this result we study
complexity of the following decision problem.\par\medskip
\noindent{\bf Problem:}\ \rm Additional minimal hypothesis (AMH)
\newline\noindent\emph{INPUT:}\ \rm Positive and negative contexts
$\mathbb{K}_+=(G_+,M,\mathcal{I}_+),\mathbb{K}_-=(G_-,M,\mathcal{I}_-)$
and a set of minimal hypotheses $\mathcal{H} = \{H_1,\ldots,H_k\}$.
\newline\noindent\emph{QUESTION:}\ \rm Is there an \emph{additional} minimal hypothesis
$H$ of $\mathbb{K}_{\pm}$ i.e. minimal hypothesis $H$ such that
$H\notin \mathcal{H}$.

\begin{algorithm}                      
\caption{FindNewMinH($\mathbb{K}_+, \mathbb{K}_-, \mathcal{H}$)}
\label{mhe_alg}                           
\begin{algorithmic}[1]
    \REQUIRE DecideAMH($\mathbb{K}_+, \mathbb{K}_-, \mathcal{H}$) = \TRUE
    \FOR {$g \in G_+$}
        \STATE $G^g_+ \Leftarrow \{g^+\cap h^+\mid h \in G_+\}$
        \STATE $I^g_+ \Leftarrow \{(g,m)\mid m \in g, g \in G^g_+\}$
        \STATE $G^g_- \Leftarrow \{g^+\cap h^-\mid h \in G_-\}$
        \STATE $I^g_- \Leftarrow \{(g,m)\mid m \in g, g \in G^g_-\}$
        \STATE $\mathbb{K}^g_+ \Leftarrow \mathbb{K}(G^g_+,M\cap g^+, I^g_+)$
        \STATE $\mathbb{K}^g_- \Leftarrow \mathbb{K}(G^g_-,M\cap g^+, I^g_-)$
        \STATE $\mathcal{H}^g \Leftarrow \{h \mid h \subseteq g^+, h\in \mathcal{H}\}$
        \IF {DecideAMH($\mathbb{K}^g_+, \mathbb{K}^g_-, \mathcal{H}^g$)}
            \RETURN FindNewMinH($\mathbb{K}^g_+, \mathbb{K}^g_-, \mathcal{H}^g$)
        \ENDIF
    \ENDFOR
    \RETURN $M$
\end{algorithmic}
\end{algorithm}

\begin{lem}\label{AMHnMHE}
AMH is in $P$ iff MHE can be solved in output polynomial time.
\end{lem}

\begin{proof}
($\Leftarrow$) Assume there is an output polynomial algorithm
$\mathcal{A}$ that generates all minimal hypotheses in
time $p(|G_+|,|M|,|\mathcal{I}_+|,|G_-|,|\mathcal{I}_-|,N)$, where
$N$ is the number of minimal hypotheses. Use this algorithm to construct
$\mathcal{A}'$ that makes first
$p(|G_+|,|M|,|\mathcal{I}_+|,|G_-|,|\mathcal{I}_-|,k + 1)$ steps of
$\mathcal{A}$. Clearly,  if there is more than $k$ minimal
hypotheses, then $\mathcal{A}'$ generates $k+1$ minimal hypotheses,
hence we can solve AMH in polynomial time.

($\Rightarrow$) Now suppose there is a function DecideAMH~$(\mathbb{K}_+, \mathbb{K}_-, \mathcal{H})$
that solves AMH problem instance in time $O(t)$. We can use \emph{Algorithm~\ref{mhe_alg}} to find an additional minimal hypothesis if there is one.
Clearly {\verb line  2} to {\verb line  8} can be computed in time $O((|G_+| + |G_-|)|M|)$. Also note that the
total number of recursive calls can not be greater than $|M|$. Thus, time complexity
of the \emph{Algorithm~\ref{mhe_alg}} is $O((|G_+| + |G_-|)|M|^2t)$. Let us prove the correctness.
First note that since hypotheses are closed in $\mathbb{K}_+$ the additional minimal hypothesis
must be a subset of some $g^+, g \in G_+$, or it could be $M$. By definition the context $\mathbb{K}^g_+$
defines exactly all closed sets of $\mathbb{K}$ that are subsets of $g^+$.
It remains to note that at the last recursive call of \emph{Algorithm~\ref{mhe_alg}}
DecideAMH($\mathbb{K}^g_+, \mathbb{K}^g_-, \mathcal{H}^g$) does not hold for any $g\in G_+$.
Thus, the only possible additional minimal hypothesis that can be returned is $M$.
\end{proof}

 Now we prove $NP$-completeness of AMH through the reduction of the most known $NP$-complete problem
 -- satisfiability of CNF -- to AMH.
\par\medskip
\noindent{\bf Problem:}\ \rm CNF satisfiability (SAT)
\newline\noindent\emph{INPUT:}\ \rm A Boolean CNF formula $f(x_1,\ldots, x_n)=C_1\wedge \ldots
\wedge C_k$
\newline\noindent\emph{QUESTION:}\ \rm Is $f$ satisfiable?
\newline
\newline
Consider an arbitrary CNF instance $C_1,\ldots,C_k$ with variables
$x_1,\ldots,x_n$, where $C_i=(l_{i_1}\vee\ldots\vee l_{ir_i}),1\leq i
\leq k$ and $l_{ij} \in \{x_1,\ldots,
x_n\}\cup\{\neg{x_1},\ldots,\neg{x_n}\}$ ($1\leq i\leq k$, $1\leq
j\leq r_i$) are literals, i.e., variables or their negations.
From this instance we construct a positive context ${\mathbb
K}_+=(G_+, M, \mathcal{I}_+)$ and a negative context ${\mathbb
K}_-=(G_-, M, \mathcal{I}_-)$ . Define
$$M=\{C_1,\ldots, C_k\}\cup\{x_1,\neg{x_1},\ldots,x_n,\neg{x_n}\}
$$ $$G_+=\{g_{x_1},g_{\neg{x_1}},\ldots,g_{x_n},g_{\neg{x_n}}\}\cup\{g_{C_1},\ldots,g_{C_k}\}$$
$$G_-=\{g_{l_1},\ldots,g_{l_n}\}$$

The incidence relation of the positive context is defined by
$\mathcal{I}_+=\mathcal{I_C}\cup\mathcal{I}_{\neq}\cup\mathcal{I}_=$,
where $$\begin{aligned} \mathcal{I_C} &= \{(g_{x_i},C_j)\mid x_i
\notin C_j, 1\leq i\leq n, 1 \leq j\leq k\}\\&\cup
\{(g_{\neg{x_i}},C_j)\mid \neg{x_i} \notin C_j, 1\leq i\leq n, 1
\leq j\leq k\}\end{aligned}$$
$$\begin{aligned}\mathcal{I}_{\neq} &= \{g_{x_1},g_{\neg{x_1}},\ldots,g_{x_n},g_{\neg{x_n}}\}\times\{x_1,\neg{x_1},
\ldots,x_n,\neg{x_n}\}\\&-\{(g_{x_1},x_1),(g_{\neg{x_1}},\neg{x_1}),\ldots, (g_{x_n},x_n),(g_{\neg{x_n}},\neg{x_n})\}\end{aligned} $$
$$\mathcal{I}_{=} = \{(g_{C_1},C_1),\ldots, (g_{C_k},C_k)\} $$
 that is for $i$-th clause $C_i^+
\cap\{g_{x_1},g_{\neg{x_1}},\ldots,g_{x_n},g_{\neg{x_n}}\}$ is the
set of literals not included in $C_i$, $\mathcal{I_{\neq}}$ is
the relation of contranominal scale.

The incidence relation of the negative context is given by
$\mathcal{I}_-=\mathcal{I_L}$ where
$$\begin{aligned}\mathcal{I_L} &= G_-\times
\{x_1,\neg{x_1},\ldots,x_n,\neg{x_n}\}\\ &-
\{(g_{l_1},x_1),(g_{l_1},\neg{x_1}),\ldots,(g_{l_n},x_n),(g_{l_n},\neg{x_n})\}\end{aligned}
$$.
\newline
\begin{center}
\begin{tabular}{rc}
    \begin{tabular}{c}
    ${\mathbb K}_+$\\
    \\
    \\
    \\
    \\
    \\
    \\
    \\
    \\
    ${\mathbb K}_-$\\
    \end{tabular}
    &
    \begin{tabular}{|l|c|c|}
    \hline
    & $C_1$~$C_2$~$\cdots$~$C_k$&~$x_1$~$\neg{x_1}$~$\cdots$~$x_n$~$\neg{x_n}$ \\
    \hline\hline
    $g_{x_1}$& &\\
    $g_{\neg{x_1}}$& &\\
    $\vdots$& $\cal{I_C}$ & $\cal{I_{\neq}}$\\
    $g_{x_n}$& &\\
    $g_{\neg{x_n}}$& &\\
    \hline
    $g_{C_1}$& &\\
    $\vdots$& $\cal{I_{=}}$ &\\
    $g_{C_k}$& &\\
    \hline
    \hline
    $g_{l_1}$& &\\
    $\vdots$& & $\cal{I_L}$\\
    $g_{l_n}$& &\\
    \hline
    \end{tabular}
    \\
\end{tabular}
\end{center}
\par\bigskip

As the set of minimal hypotheses we take
$\mathcal{H}=\{\{C_1\},\{C_2\},\ldots,\{C_k\}\}$. It is easy to see
that ${\mathbb K}_{\pm}$ with $\mathcal{H}$ is a correct instance of
AMH.

If a hypothesis (not necessary minimal)  is not contained in
$\mathcal{H}$ we will call it \emph{additional}.

\begin{prop}\label{P1}
If $H$ is an additional minimal hypothesis of ${\mathbb K}_{\pm}$ then
\newline
$H\subseteq \{x_1,\neg{x_1},\ldots,x_n,\neg{x_n}\}$.
\end{prop}
\begin{proof} Suppose $H \nsubseteq
\{x_1,\neg{x_1},\ldots,x_n,\neg{x_n}\}$, then since $H$ is not empty
there is some $C_i \in H$, $1\leq i\leq k$. But $H$ is a minimal
hypothesis and thus it does not contain any hypothesis. Hence
$H=C_i$ and this contradicts the fact that $H$ is an \emph{additional}
minimal hypothesis.
\end{proof}

For any $H \subseteq \{x_1,\neg{x_1},\ldots,x_n,\neg{x_n} \}$ that
satisfies $\{x_i,\neg{x_i}\}\cap H \neq \emptyset$ for any $1\leq i\leq n$ we
define the truth assignment $\varphi_H$ in a natural way:
$$\varphi_H(x_i)=\begin{cases}
true,&\text{if $x_i\in H$;}\\
false,&\text{if $x_i\notin H$;}
\end{cases}
$$
In the case $\{x_i,\neg{x_i}\} \cap H = \emptyset$ for some $1\leq i\leq
n$, $\varphi_H$ is not defined. We define $\varphi_H(x_i)=true$ even if
$\{x_i,\neg{x_i}\} \subseteq H$, although in this case it can be defined by eigther way.

Symmetrically, for a truth assignment $\varphi$ define the set $H_{\varphi}
= \{x_i\mid \varphi(x_i) = true\} \cup \{\neg{x_i}\mid \varphi(x_i) =
false\}$.

Below, for $H
\subseteq\{x_1,\neg{x_1},\ldots,x_n,\neg{x_n}\}$ we will denote the
complement of $H$ in $\{x_1,\neg{x_1},\ldots,x_n,\neg{x_n}\}$ by
$\overline{H}$.

\begin{prop}\label{P2}
If a subset $H \subseteq \{x_1,\neg{x_1},\ldots,x_n,\neg{x_n}\}$ is
not contained in the intent of any negative example (i.e $\forall g \in
G_-,H\nsubseteq g^-$), then $\varphi_{\overline{H}}$ is
defined. Conversely, for a truth assignment $\varphi$ the set
$\overline{H_{\varphi}}$ is not contained in the intent of any negative
concept.
\end{prop}
\noindent The proof is straightforward.
\newline
\newline
The following theorem proves NP-hardness of AMH.

\begin{thm}
  AMH has a solution if and only if SAT has a solution.
\end{thm}
\begin{proof}$(\Rightarrow)$ Let $H$ be an additional
minimal hypothesis of ${\mathbb K}_{\pm}$. First note that by
Proposition \ref{P1} and Proposition \ref{P2} the truth assignment
$\varphi_{\overline{H}}$ is correctly defined. Since $H$ is a nonempty
concept intent of $\mathbb{K}_+$, Proposition~\ref{P1} together with the fact
that $I_{\neq}$ is the relation of contranominal scale implies
$H^+=\{g_{x_i}\mid x_i \in \overline{H}\} \cup \{g_{\neg{x_i}}\mid
\neg{x_i} \in \overline{H}\}$. Now
$H^{++}\cap\{C_1,C_2,\ldots,C_k\}=\emptyset$, hence for any $C_i$
($1\leq i \leq k$) there is some $g_l \in H^+$ such that $g_l \notin
C_i^+$. According to the definition of $\mathcal{I_C}$ the latter
means that literal $l$ belongs to clause $C_i$. Thus
$f(\varphi_{\overline{H}})=true$.

($\Leftarrow$) Let $\varphi$ be a truth assignment and $f(\varphi)=true$.
Define $H=\overline{H_{\varphi}}$. Note that $H^+=\{g_{x_i}\mid x_i \in
H_{\varphi}\}\cup \{g_{\neg{x_i}}\mid \neg{x_i} \in H_{\varphi}\}$, because
$\mathcal{I}_{\neq}$ is the relation of contranominal scale and
$H\cap g_{C_j}^+=\emptyset,1\leq i \leq k$. Suppose that $C_i \in
H^{++}$ for some $1\leq i \leq k$. This is equivalent to $H^+
\subseteq C_i^+$. Hence, by definition of $\mathcal{I_C}$, there is no
literal $l\in H_{\varphi}$ such that $l \in C_i$. Therefore, the clause
$C_i$ does not hold and this contradicts the fact that $\varphi$ satisfies CNF
$f$. Thus $H^{++}=H$ and $H$ is a hypothesis. Since $H$ does not
contain any $\{C_i\}$, it must contain an additional minimal
hypothesis.
\end{proof}

\begin{cor}\label{cor:MHE}
  MHE cannot be solved in output polynomial time, unless $P=NP$.
\end{cor}

\section{Dualizing monotone Boolean functions on lattices}\label{sec:monfunc}
Let $f$ be a monotone Boolean function on a lattice $\mathcal{L}$. Without loss of generality we can assume
that $\mathcal{L}$ is a concept lattice $\mathcal{L}=\mathfrak{B}(G,M,I)$ of the
corresponding formal context $\mathbb{K}=(G,M,I)$. Then $A \subseteq
B \Rightarrow f((A,A')) \leq f((B,B'))$. It is known that any
monotone Boolean function on a lattice is uniquely given by its
minimal 1-values, i.e. by the set $\mathcal{A}=\{(A,A')\mid (A,A') \in
\mathfrak{B}, f((A,A'))=1, f((B, B'))=0\ \forall B\subset A \}$.
Define positive context $\mathbb{K}_+ = \mathbb{K}$. Define negative context $\mathbb{K}_-=(G_-,M, I_-)$
via its set of objects intents $G_-=\{g_A\mid (A', A) \in \mathcal{A}\}$ and ${g_A}^{-} = A$. In other words negative examples are precisely intents of minimal 1-values of $f$. Clearly set of minimal hypotheses of $\mathbb{K}_{\pm}$ is exactly the set of maximal 0-values of $f$.

Symmetrically, for a given positive and negative contexts $\mathbb{K}_+$ and $\mathbb{K}_-$
define context $\mathbb{K}_{+\cup\_}=(G_+\cup G_-, M, I_+ \cup I_-)$. Let $f$ be a monotone
Boolean function on $\mathbb{K}_{+\cup\_}$ that is given by its minimal 1-values $\mathcal{A}=\{({g^-}', g^-)\mid g \in G_-\}$ ($(\cdot)'$ -- derivation operator of $\mathbb{K}_{+\cup\_}$). It is not hard to see that the set of maximal 0-values of $f$ is defined by the set of minimal hypotheses of $\mathbb{K}_{\pm}$.

From Corollary~\ref{cor:MHE} it follows that the following problem cannot be solved in output
polynomial time unless $P=NP$ \par\medskip
\noindent{\bf Problem:}\ \rm Maximial false values enumeration (MFE)
\newline\noindent\emph{INPUT:}\ \rm
A formal context $\mathbb{K}$ and a set of minimal 1 values of
monotone Boolean function $f$ on the concept lattice of $\mathbb{K}$.
\newline\noindent\emph{OUTPUT:}\ \rm Set of maximal 0 values of $f$.
\par\medskip

Lemma~\ref{AMHnMHE} also implies that the dualization problem on a lattice given by a formal context can be solved in output polynomial time iff
the corresponding duality (decision version of dualization) problem can be solved in polynomial time.

Note that in the case of Boolean lattice MFE problem is polynomially equivalent to Monotone Boolean
Dualization and minimal 0 values in this case can
be enumerated in quasi-polynomial time $O(N^{o(\log{N})})$,
where $N$ is $|input\ size|+|output\ size|$ (see~\cite{fk96}).


In database theory a closure of a set of attributes $A$ is defined by means of iterated applications of functional dependencies with premises contained in $A$. Same type of closure, by means of implications instead of functional dependencies, is known in FCA. More precisely, applying
imp$(A) = A\cup \{B \mid D\to B, D\subseteq A\}$ iteratively to $A$ by putting at each next step $A:\colon =$~imp$(A)$ until saturation, one obtains implicational closure of $A$, which is equal to $A''$~\cite{gw}. So, the set of all implications of a context defines the closure operator $(\cdot)''$, closed subsets of attributes,  which together with the respective closed subsets of objects (extents) give the concept lattice. Hence, instead of defining a lattice by the ordered set of its irreducible elements, one can define it in terms of the set of all valid implications of the respective formal context, or, equivalently, by its implication base. This consideration poses another setting of the dualization problem, where the lattice -- instead of the set of positive examples $G_+$ -- is given by its implications or implication base, and one has to dualize the monotone function given by the set of examples $G_-$.
When the lattice is Boolean, its implication base is empty~\cite{gw}, so one has to dualize the set of examples $G_-$, which can be considered as a monotone DNF, where disjunction goes over objects -- elements of $G_-$ -- which themselves can be taken as conjunctions of the respective attributes. When the lattice is distributive, its minimum implication base has one-element premises~\cite{gw} (hence, the number of implications in the base is not larger than $|M|$), so it can easily be computed from the context in polynomial time, and vice versa. Therefore, the dualization on lattices given by implication bases for distributive lattices is polynomially equivalent to the dualization on lattices given by contexts (ordered sets of irreducible elements), which we study in the next section.
The study of dualization problems for lattices given by implication bases is motivated by simple linear-time reciprocal translations of implications to functional dependencies~\cite{ko08} and propositional Horn theories~\cite{emg08}.

In~\cite{kss00} it has been proven that the following problem is NP-hard:
\par\medskip
\noindent{\bf Problem:}\ \rm Incremental maximal model (IME)
\newline\noindent\emph{INPUT:}\ \rm Horn theory $\Phi$ and a set of its maximal models $S$.
\newline\noindent\emph{QUESTION:}\ \rm Is there another maximal model of $\Phi$ not contained in $S$?
\par\medskip
In terms of FCA a Horn theory corresponds to a set of implications $\mathcal{J}$
and maximal models correspond to inclusion maximal closed sets of $\mathcal{J}$, or object intents, that are not $M$.
In the dualization setting maximal closed sets
are dual to the singleton set $\{M\}$. Hence for the

\par\medskip
\noindent{\bf Problem:}\ \rm Minimal true values enumeration, on lattice given by implication base (MTEIB)
\newline\noindent\emph{INPUT:}\ \rm
A lattice $\mathcal{L}(\mathcal{J})$ given by an implication base $\mathcal{J}$
and a set of maximal 0 values of
monotone Boolean function $f$ on the lattice $\mathcal{L}(\mathcal{J})$.
\newline\noindent\emph{OUTPUT:}\ \rm Set of minimal 1 values of $f$.
\par\medskip

we have the following

\begin{cor}
A solution of MTEIB is impossible in output polynomial time unless $P=NP$.
\end{cor}

\section{Dualization and minimum implication bases}

In this section we give complexity upper bounds of some important special cases of monotone Boolean dualization on lattices in terms of the complexity of minimum implication base problem (i.e. minimum Horn theory).

\noindent{\bf Problem:}\ \rm Minimum implication base recognition (MIBR)
\newline\noindent\emph{INPUT:}\ \rm Formal context $\mathbb{K}=(G,M,I)$, set of implications $\mathcal{J}$.
\newline\noindent\emph{QUESTION:}\ \rm Is $\mathcal{J}$ implication base of $\mathbb{K}$?

The complexity of MIBR problem is a long standing open problem. The only known complexity result is that MIBR is at least hard as monotone Boolean duality \cite{kh95,emg08}.

As we have shown monotone Boolean duality on a lattice given by a formal context is coNP-complete. It turns out that if we additionally have an implication base as input then the problem does not get harder than MIBR.

\noindent{\bf Problem:}\ \rm Duality over lattices given by formal context and implication base (DCI)
\newline\noindent\emph{INPUT:}\ \rm formal context $\mathbb{K}=(G,M,I)$, antichains
$\mathcal{A}, \mathcal{B} \subseteq \mathcal{L}(\mathbb{K})$ satisfying (*), implication base $\mathcal{J}$ of $\mathcal{L}(\mathbb{K})$.
\newline\noindent\emph{QUESTION:}\ \rm Are $\mathcal{A}$ and $\mathcal{B}$ dual on $\mathcal{L}(\mathbb{K})$?

Note that $\mathcal{J}$ could be any implication base of $\mathbb{K}$ that is not necessary minimum. Now we describe polynomial (Karp-)reduction of DCI to MIBR. Let us define a context
$\mathbb{K}_{\mathcal{B}} = (G_{\mathcal{B}}, M, I_{\mathcal{B}})$, where
$G_{\mathcal{B}}=\{g_B\mid g \in G, B \in \mathcal{B}\}$ ($|G_{\mathcal{B}}| = |G|\times|\mathcal{B}|$), and relation $I_{\mathcal{B}}$ is defined via object intents $g_B' = g' \cap B$ for any $g_B \in G_{\mathcal{B}}$.
Obviously, a set $X \subseteq M$ is closed in $\mathbb{K}_{\mathcal{B}}$ iff $X$ is closed in $\mathbb{K}$ and there is $B \in \mathcal{B}$ that $X\subseteq B$. Define implication base $\mathcal{J}_{\mathcal{A}} = \mathcal{J} \cup \{A\rightarrow M\mid A \in \mathcal{A}\}$. Clearly, a set $X$ is closed (satisfied) in $\mathcal{J}_{\mathcal{A}}$  iff $X = M$ or $X$ is closed in $\mathbb{K}$ and $A \nsubseteq X$ for any $A \in \mathcal{A}$. Thus $\mathcal{A}$ and $\mathcal{B}$ are dual on $\mathcal{L}(\mathbb{K})$
iff $\mathcal{J}_{\mathcal{A}}$ is implication base of $\mathbb{K}_{\mathcal{B}}$. We have proven:

\begin{lem}\label{MIBHARD}
MIBR is DCI-hard (under polynomial Karp-reduction)
\end{lem}

In \cite{elb02,el09} the problem of dualization over product of lattices was considered. For the case of semi-lattices of bounded width Elbassioni has shown that the duality problem can be solved in quasi-polynomial time. Nevertheless in case of product of general lattices the existence of quasi-polynomial algorithm is still an open problem. Here we prove that this problem is not harder than MIBR.

\noindent{\bf Problem:}\ \rm Duality over product of lattices (DPL)
\newline\noindent\emph{INPUT:}\ \rm Product of lattices $\mathcal{L}=\mathcal{L}_1\times \ldots \times \mathcal{L}_k$ given by $\mathcal{L}_1, \ldots, \mathcal{L}_k$, antichains $\mathcal{A}, \mathcal{B} \subseteq \mathcal{L}$ satisfying (*),
\newline\noindent\emph{QUESTION:}\ \rm  Are $\mathcal{A}$ and $\mathcal{B}$ dual over $\mathcal{L}$?

\begin{prop}\label{PROD}
MIBR is DPL-hard (under polynomial Karp-reduction)
\end{prop}

\begin{proof}
First note that given a lattice $\mathcal{L}_i$ (e.g. as a whole relation matrix) we can find all join-irreducible and meet-irreducible elements of $\mathcal{L}_i$ in $poly(|\mathcal{L}_i|)$ time. Thus it is possible to get context $\mathbb{K}_{\mathcal{L}_i}=(G_i, M_i, I_i)$ that defines lattice $\mathcal{L}_i$ in polynomial time.
In order to construct a formal context $\mathbb{K}_{\mathcal{L}}=(G,M,I)$ of the product of lattices $\mathcal{L}$,
we define $G = G_1 \sqcup \ldots \sqcup G_k$, $M = M_1 \sqcup \ldots \sqcup M_k$, and relation $I$. Without loss of generality let $g \in G_i$ and $m \in M_j$ then  $g I m $ iff $i\neq j$ or $g I_i m$. It is straightforward to check that $\mathcal{L}(\mathbb{K})$ is isomorphic to $\mathcal{L}$.

In \cite{dp92} (Lemmas A.2 and A.3) it was proven that (in FCA terms)
{\it
Given a formal context $\mathbb{K}=(G,M,I)$ one can compute its cardinality-minimum implication base $\mathcal{J}$ in $O(|M|^2|\mathcal{L}(\mathbb{K})|^2)$ time. Moreover, such a $\mathcal{J}$ contains at most $|M|^2|\mathcal{L}(\mathbb{K})|$ implications. }
Thus for a given lattice $\mathcal{L}_i$ we can find implication base $\mathcal{J}_i$ of size $O(poly(|\mathcal{L}_i|)$ in time $O(poly(|\mathcal{L}_i|))$. Clearly, $\mathcal{J}_1\cup\ldots\cup\mathcal{J}_k$ is an implication base of
$\mathbb{K}_{\mathcal{L}}$.
    The proposition statement follows from Lemma~\ref{MIBHARD}.
\end{proof}

Another interesting special case of lattices for which we can establish similar complexity bound is the case of distributive lattices.
It is known that for a given context $\mathbb{K}$ of a distributive lattice, the minimum implication base of $\mathbb{K}$ has size
polynomial in $|\mathbb{K}|$ and can be found in polynomial time (\cite{gw}). Thus MIBR is in P for a distributive lattice.
The following Corollary is directly implied from this fact and Lemma~\ref{MIBHARD}:

\begin{cor}
Dualization on distributive lattice problem:
Given formal context $\mathbb{K}$ of a distributive lattice and antichains
$\mathcal{A},\mathcal{B} \subseteq \mathcal{L}(\mathbb{K})$ satisfying (*),
decide whether $\mathcal{A}$ and $\mathcal{B}$ are dual or not?
Is not harder than MIBR (under polynomial Karp-reduction).
\end{cor}

\section{Dualization over distributive lattices}

We assume that a distributive lattice is represented as a lattice $\mathcal{L}(\mathcal{P})$ of
downsets (order ideals) of a poset $\mathcal{P}$, and poset $\mathcal{P}$
is given by a matrix $n \times n$. It is well known that any distributive lattice has such a representation~\cite{dp,g11,gw}.
Note that one can use formal context representation of the distributive lattice as well, since
the size of the corresponding formal context $(P,P,\leq)$ is polynomial in $n$, and our dualization algorithm is subexponential.

We treat the elements of $\mathcal{L}=\mathcal{L}(\mathcal{P})$ as subsets of $\mathcal{P}$ (since they are downsets of
$\mathcal{P}$), so for two downsets $A, B \in \mathcal{L}(\mathcal{P})$ $A \leq B$ means that $A \subseteq B$.
For an element $p \in \mathcal{P}$, the smallest (by set inclusion) downset that contains $p$ is denoted by $\da{p}$, and the smallest upperset (order filter) that contains $p$ is denoted by $\ua{p}$. More generally, for any subset $X \subseteq \mathcal{P}$, by $\da X$ we denote the smallest downset that contains $X$, i.e. $\da X=\cup_{p\in X} \da p$.

Let $\mathcal{A}$ and $\mathcal{B}$ be antichains of a distributive lattice $\mathcal{L(P)}$.
Further on we will call a triple of the form ($(\mathcal{A}, \mathcal{B}), \mathcal{P}$) dualization problem input.
Note that in the degenerate cases where $\mathcal{A}=\emptyset$
or $\mathcal{B}=\emptyset$ the duality can easily be tested in polynomial time. If $\mathcal{A}$ is empty, then $\mathcal{B}$ is dual to $\mathcal{A}$ iff $\mathcal{B} =\{\mathcal{P}\}$. If $\mathcal{B}$ is empty, then $\mathcal{A}$ is dual to $\mathcal{B}$ iff
$\mathcal{A} =\{\emptyset\}$. Let us call the algorithm that tests duality in these two degenerate cases \emph{EasyTest((A,B),P)}.
\par\medskip
\noindent
We will also use the notion of frequency of an element $p\in\mathcal{P}$.
Let $\mathcal{C}$ be some set of subsets of $\mathcal{P}$ (i.e. $\mathcal{C} \subseteq 2^{\mathcal{P}}$), then
the frequency of $p$ in $\mathcal{C}$ is the fraction of elements of $\mathcal{C}$ that contain $p$:
\begin{defn}
$freq_{\mathcal{C}}(p) = |\{C \in \mathcal{C}\mid p \in C\}| / |\mathcal{C}|$.
\end{defn}

Let us denote $\overline{\mathcal{C}}=\{\mathcal{P}\setminus C \mid C \in \mathcal{C}\}$, thus by definition
$freq_{\overline{\mathcal{C}}}(p) = |\{C \in \mathcal{C}\mid p \notin C\}| / |\mathcal{C}|$.

For convenience we define the quantities $N = |\mathcal{A}| + |\mathcal{B}|$, and
$m = \max_{p\in \mathcal{P}}{(|\da{p}|+|\ua{p}|)}$ (note that $m \geq 2$).

\subsection{Algorithm}
\label{sect:algorithm}

Here we describe a subexponential algorithm for testing duality on a distributive lattice.
The structure of the algorithm is close to that in \cite{fk96}. The algorithm decomposes the initial problem instance
into smaller instances and solves them recursively. In order to keep the total number of recursive calls
subexponential at each decomposition step, the algorithm tries to select an element of $\mathcal{P}$
such that either it is \emph{frequent} or it has a \emph{large} fraction of successors of predecessors.

\begin{algorithm}                      
\caption{TestDuality($(\mathcal{A}, \mathcal{B}), \mathcal{P}$)}
\label{distr_latt_alg}                           
\begin{algorithmic}[1]                   
    \REQUIRE $\mathcal{A},\mathcal{B} \subseteq \mathcal{L}(\mathcal{P})$
    \IF{$\mathcal{A} = \emptyset$ or $\mathcal{B} = \emptyset$}
      \RETURN EasyTest$((\mathcal{A}, \mathcal{B}), \mathcal{P})$
    \ENDIF
    \STATE $n \Leftarrow |\mathcal{P}|$
    \STATE $m \Leftarrow \max_{p\in \mathcal{P}}{(|\da{p}|+|\ua{p}|)}$
    \STATE $N = |\mathcal{A}| + |\mathcal{B}|$
    \IF{$m > n^{1/3}$}
      \STATE $p \Leftarrow \argmax_{p\in \mathcal{P}}{(|\da{p}|+|\ua{p}|)}$
    \ELSE
      \IF{$\max_{p\in \mathcal{P}}{freq_{\mathcal{A}}(p)} < \frac{1}{m \log_{4/3}{N}}$
          \AND $\max_{p\in \mathcal{P}}{freq_{\overline{\mathcal{B}}}(p)} < \frac{1}{m^2 \log_{4/3}{N}}$}
          \RETURN \FALSE
      \ENDIF
      \STATE $p \Leftarrow \argmax_{p\in \mathcal{P}}{(
            \max(freq_{\mathcal{A}}(p)},
                 freq_{\overline{\mathcal{B}}}(p)))$
    \ENDIF
    \RETURN TestDuality($({\mathcal{A}}_1^p, {\mathcal{B}}_1^p), \mathcal{P} \setminus \da{p}$) $\wedge$
            TestDuality($({\mathcal{A}}_2^p, {\mathcal{B}}_2^p), \mathcal{P} \setminus \ua{p}$)
\end{algorithmic}
\end{algorithm}

To describe decomposition performed by our algorithm we define the following four sets:
\begin{center}$\mathcal{A}^p_1 = \{A \setminus \da{p} \mid A \in \mathcal{A} \},\
\mathcal{B}^p_1 = \{B \setminus \da{p} \mid p \in B,\ B \in \mathcal{B} \},$ \end{center}
\begin{center}$\mathcal{A}^p_2 = \{A  \mid p \notin A,\ A \in \mathcal{A} \},\
\mathcal{B}^p_2 = \{B \setminus \ua{p} \mid B \in \mathcal{B} \}.$ \end{center}

Note that $\mathcal{B}^p_1 = \{B \setminus \da{p} \mid \da{p} \subseteq B,\ B \in \mathcal{B} \}$,
and $\mathcal{A}^p_2 = \{A  \mid \ua{p} \cap A = \emptyset,\ A \in \mathcal{A} \}$.

\noindent The following lemma proves the correctness of \emph{Algorithm~\ref{distr_latt_alg}}.

\begin{lem}
  For any $p\in \mathcal{P}$, $\mathcal{A}$ and $\mathcal{B}$ are dual iff the following two conditions hold:
  \newline
  $\mathcal{A}^p_1$ and $\mathcal{B}^p_1$ are dual on $\mathcal{L}(\mathcal{P} \setminus \da{p})$,\\
  $\mathcal{A}^p_2$ and $\mathcal{B}^p_2$ are dual on $\mathcal{L}(\mathcal{P} \setminus \ua{p})$
\end{lem}

\begin{proof}

($\Leftarrow$) Let us fix arbitrary $X \in \mathcal{L}$. Consider two possible cases: $p \in X$ and $p \notin X$.
If $p \in X$ then since $\mathcal{A}^p_1$ and $\mathcal{B}^p_1$ are dual,
either $A_1 \subseteq X\setminus \da{p}$ for some $A_1 \in \mathcal{A}^p_1$, or
$X \setminus \da{p} \subseteq B_1$ for some $B_1 \in \mathcal{B}^p_1$.
Clearly, $X \setminus \da{p} \subseteq B_1$ implies $X \subseteq B_1\cup \da{p} \in \mathcal{B}$.
On the other hand $A_1 \in \mathcal{A}^p_1$ implies that there is
$A\in \mathcal{A}$ such that $A_1 = A\setminus \da{p}$, and hence $A \subseteq X$ (since $\da{p} \subseteq X$).

If $p \notin X$ then since $\mathcal{A}^p_2$ and $\mathcal{B}^p_2$ are dual
either $A_2 \subseteq X\setminus \ua{p}$ for some $A_2 \in \mathcal{A}^p_2$, or
$X \setminus \ua{p} \subseteq B_2$ for some $B_2 \in \mathcal{B}^p_2$.
By definition $B_2 \in \mathcal{B}^p_2$ implies that there is
$B\in \mathcal{B}$ such that $B_2 = B\setminus \ua{p}$.
Note that $A_2 \in \mathcal{A}$, and $X = X \setminus \ua{p} \subseteq B_2 \subseteq B$.

($\Rightarrow$) Let us prove that $\mathcal{A}^p_1$ and $\mathcal{B}^p_1$ are dual. Consider
arbitrary $X \in \mathcal{L}(\mathcal{P} \setminus \da{p})$. Because $\mathcal{A}$ and  $\mathcal{B}$
are dual on $\mathcal{L}(\mathcal{P})$ either $A \subseteq X\cup \da{p}$ for some $A \in \mathcal{A}$,
or $X\cup \da{p} \subseteq B$ for some $B \in \mathcal{B}$. If $A \subseteq X\cup \da{p}$ then
$A \setminus \da{p} \subseteq X$ (since $\da{p} \cap X = \emptyset$). If $X\cup \da{p} \subseteq B$ then
$X \subseteq B \setminus \da{p}$, and by definition $B \setminus \da{p} \in \mathcal{B}^p_1$. It is easy to check that $(\mathcal{A}^p_1, \mathcal{B}^p_1)$ has property $(*)$

Now we prove that $\mathcal{A}^p_2$ and $\mathcal{B}^p_2$ are dual. Consider
arbitrary $X \in \mathcal{L}(\mathcal{P} \setminus \ua{p})$. Note that $X \in \mathcal{L}(\mathcal{P})$.
Because $\mathcal{A}$ and  $\mathcal{B}$ are dual on $\mathcal{L}(\mathcal{P})$ either $A \subseteq X$
for some $A \in \mathcal{A}$, or $X \subseteq B$ for some $B \in \mathcal{B}$. If $A \subseteq X$ then
$p \notin A$, and $A \in \mathcal{A}^p_2$. If $X \subseteq B$ then
$X \subseteq B \setminus \ua{p}$ (since $\ua{p} \cap X = \emptyset$).
It is easy to check that $(\mathcal{A}^p_2, \mathcal{B}^p_2)$  has property $(*)$.

\end{proof}

The following lemma helps one to establish a lower bound on the frequency of the most frequent element of $\mathcal{P}$.

\begin{lem}\label{inequality_lemma}
  If $\mathcal{A}$ and $\mathcal{B}$ are dual then
  $$ \sum_{A\in \mathcal{A}}{{(3/4)}^{|A|/m^2}} + \sum_{B\in \mathcal{B}}{e^{-(n-|B|)/m}} \geq 1 $$
\end{lem}
\begin{proof}
To prove this bound we use the 'method of expectations' similar to that in \cite{fk96}, but with a more tricky probability distribution.
Suppose we fixed some probability distribution of $X \in \mathcal{L}$.
Let us denote the expected number of $A\in \mathcal{A}, A\subseteq X$ by $E_{\mathcal{A}}$, and the expected number of
$B\in \mathcal{B}, X\subseteq B$ by $E_{\mathcal{B}}$.
Antichains $\mathcal{A}$ and $\mathcal{B}$ are dual iff for any $X \in \mathcal{L}$ either $A\subseteq X$, for some
$A \in \mathcal{A}$, or $X\subseteq B$, for some $B \in \mathcal{B}$. Thus if $\mathcal{A}$ and $\mathcal{B}$ are dual, then
$E_{\mathcal{A}} + E_{\mathcal{B}} \geq 1$. By linearity of expectations
$E_\mathcal{A} = \sum_{A\in \mathcal{A}}{E_A}$, where $E_A$ is probability that $A \subseteq X$.
Similarly, $E_\mathcal{B} = \sum_{B\in \mathcal{B}}{E_B}$, where $E_B$ is the probability that $X \subseteq B$.
Unlike to the case of Boolean lattice, no analytical expression for $E_A$ and $E_B$ is known
(even the existence of a polynomial approximation algorithm is an open question~\cite{dggj04}),
but we can find upper bounds for
$E_A,\ A \in \mathcal{A}$ and $E_B,\ B \in \mathcal{B}$.

In order to generate random (but not uniform) element $X \in \mathcal{L}$ we select each $p \in \mathcal{P}$
with probability $1/m$. Suppose elements $p_1,p_2,\ldots,p_r$ have been selected, then the resulting downset
$X \in \mathcal{L}$ is defined as $X = \da{p_1} \cup \da{p_2} \cup \ldots \cup \da{p_r}$.

For a given downset $A\in \mathcal{A}$ let us bound the probability that $A\subseteq X$.
To each $p\in\mathcal{P}$ we assign an event $I_p$ such that $p \in X$.
Note that $Pr(\overline{I_p})\geq (1-1/m)^m\geq 1/4$ (since $m \geq 2$). Consider any maximum-cardinality set $\{a_1,a_2,\ldots,a_k\} \subseteq A$ such that events $I_{a_1},I_{a_2},\ldots,I_{a_k}$ are mutually independent. For any $a\in A$ event $I_a$ happens only if some $q \geq a$ was selected, hence $I_a$ is independent of all $I_q$ for $q \notin \da{(\ua{a})}$. Since
$|\da{(\ua{a})}|\leq m^2$ it is easy to see that $k\geq |A|/m^2$.
Since event $A \subseteq X$ happens if $I_{a_1} \wedge I_{a_1} \wedge \ldots \wedge I_{a_k}$
we have $Pr(A\subseteq X) \leq \prod_{1\leq i \leq k}{(1-Pr(\overline{I_{a_i}}))} \leq (1-1/4)^{|A|/m^2}$.

To bound $E_B$, note that for any $B \in \mathcal{B}$, the probability $Pr(X \subseteq B)=
Pr(X \cap (\mathcal{P} \setminus B)=\emptyset)$. This probability is exactly
$(1-1/m)^{|\mathcal{P} \setminus B|}=(1-1/m)^{n-|B|} \leq e^{-(n-|B|)/m}$
\end{proof}

\begin{cor}\label{cor:freq}
  If $\mathcal{A}$ and $\mathcal{B}$ are dual, then at least one of the following statements is true:
  \begin{itemize}
  \item $\exists p \in \mathcal{P}: freq_{\mathcal{A}}(p) \geq \frac{1}{m \log_{4/3}{N}}$
  \item $\exists p \in \mathcal{P}: freq_{\overline{\mathcal{B}}}(p) \geq \frac{1}{m^2 \log_{4/3}{N}}$
  \end{itemize}
\end{cor}
\begin{proof}
Let $k_A = \min_{A \in \mathcal{A}}{|A|/m^2}$, $k_B = \min_{B \in \mathcal{B}}{(n-|B|)/m}$, and $k = \min(k_A,k_B)$.
By Lemma \ref{inequality_lemma}
$\sum_{A\in \mathcal{A}}{{(3/4)}^{|A|/m^2}} + \sum_{B\in \mathcal{B}}{(3/4)^{(n-|B|)/m}} \geq 1$.
Hence $(3/4)^k N \geq 1$ which yields $k \leq \log_{4/3} N$. Since $(\mathcal{A}, \mathcal{B})$ has property $(*)$, for
any $A \in \mathcal{A},\ B \in \mathcal{B}$ the intersection $A \cap \overline{B}$ is nonempty.
If $|A|=k m^2$, then there is some $a\in A$ such that
$freq_{\overline{\mathcal{B}}}(a) \geq 1/(k m^2) \geq 1/(m^2\log_{4/3}{N})$. Similarly, if
$|\overline{B}|=k m$, then there is some $b\notin B$ such that $freq_{\mathcal{A}}(b) \geq 1/(km) \geq 1/(m\log_{4/3}{N})$.
\end{proof}

\begin{thm}[Time complexity of the dualization algorithm]
  \emph{Algorithm \ref{distr_latt_alg}} decides duality in time $2^{O(n^{0.67}\log^3(|\mathcal{A}|+|\mathcal{B}|))}$.
\end{thm}

\begin{proof}
First note that all lines of \emph{Algorithm~\ref{distr_latt_alg}} can be computed in polynomial time (disregarding recursive calls).
In order to bound the number of recursive calls during an execution of \emph{Algorithm~\ref{distr_latt_alg}},
we consider the following \emph{problem volume} quantity:
$vol(\mathcal{A},\mathcal{B},\mathcal{P})=|\mathcal{A}|\cdot|\mathcal{B}|\cdot n$.
Dualization problem $(\mathcal{A},\mathcal{B},\mathcal{P})$ branches into two
subproblems $(\mathcal{A}^p_1,\mathcal{B}^p_1,\mathcal{P} \setminus \da{p})$ and
$(\mathcal{A}^p_2,\mathcal{B}^p_2,\mathcal{P} \setminus \ua{p})$. Let us denote the volumes of these problems
by $vol$, $vol_1$, and $vol_2$, respectively. In case of {\verb line  13} by Corollary~\ref{cor:freq}
either $vol_2 \leq (1-\frac{1}{m \log N})vol$ or
$vol_1 \leq (1-\frac{1}{m^2 \log N})vol$.
Moreover, in case of {\verb line  8} of the \emph{Algorithm \ref{distr_latt_alg}}, $m = |\da{p}|+|\ua{p}| > n^{1/3}$,
which implies either $vol_1\leq (n - \frac{m}{2})/n \cdot vol \leq (1-\frac{1}{2n^{2/3}})vol$,
or $vol_2 \leq (1-\frac{1}{2n^{2/3}})vol$. Thus, we have the following bound on the number of recursive calls:
$A(vol)\leq A((1-\frac{1}{2n^{2/3}logN})vol) + A(vol-1) + 1$. In \cite{fk96}
it has been proven that solution $A(v)$ of the recurrence $A(v)\leq 1 + A((1-\varepsilon)v) + A(v-1),\ A(1) = 1$ can be bounded by
$A(v)\leq (3+2v\varepsilon)^{\log v / \varepsilon}$.
Substituting $\varepsilon = \frac{1}{2n^{2/3}logN}$ yields $A(v) \leq (3+2N^2n^{1/3})^{O((\log N + \log n)n^{2/3}\log N)} \leq
2^{O((\log N + \log n)^2 n^{2/3}\log N)} \leq 2^{O(n^{0.67}\log^3N)}$.
\end{proof}

\section{Conclusion}

In this paper we have studied the dualization problem on a lattice given by the ordered sets of its irreducible elements (i.e., as a concept lattice). For this representation, the dualization problem has complexity different from that in case of explicit lattice representation as an ordered set of all its elements. We have shown that the dualization problem for a lattice given by the ordered set of its irreducible elements (concept lattice) is equivalent to the enumeration of minimal hypotheses, which is not possible in output polynomial time unless P=NP. For the case of distributive lattices dualization was shown to be possible in subexponential time.
We have proved that the long standing open complexity problem of constructing minimum implication base (irredundant Horn CNF) is at least as hard as dualization over distributive lattice or dualization over the product of explicitly given lattices (open problem stated by Elbassioni~\cite{el09}).

It is still open whether dualization over distributive lattice
can be solved in output quasi-polynomial time, or this problem cannot be solved in output polynomial time unless P = NP.
The complexity of dualization for other important classes of lattices, such as modular, also remains an open question for the case where the lattice is given by the ordered set of its irreducible elements.

\section*{Acknowledgments}

We thank Kazuhisa Makino and Lhouari Nourine for helpful discussions. The second author was supported by the Basic Research Program of the National Research University Higher School of Economics (Moscow, Russia) and Russian Foundation for Basic Research.

\end{document}